\documentclass[11pt]{article}

\usepackage[letterpaper, margin=1in]{geometry}
\usepackage{amsfonts,amssymb,amsmath,amsthm}
\usepackage[utf8]{inputenc}
\usepackage[T1]{fontenc}
\usepackage{newtxtext}
\usepackage[hidelinks]{hyperref}
\usepackage[shortlabels]{enumitem}
\usepackage{color}
\usepackage{tikz}
\usetikzlibrary{arrows,matrix,positioning}
\usepackage{tikz-cd}
\usepackage{graphicx,  fullpage, url}
\usepackage{verbatim}
\usepackage{mathrsfs}
\usepackage{mathtools}
\usepackage{array}
\usepackage{thm-restate}
\usepackage{setspace} 
\theoremstyle{plain}
\newtheorem{thm}{Theorem}[section]
\newtheorem{lem}[thm]{Lemma}
\newtheorem{defn}[thm]{Definition}
\newtheorem*{defn*}{Definition}
\newtheorem{cor}[thm]{Corollary}
\newtheorem{claim}[thm]{Claim}

\theoremstyle{remark}
\newtheorem{rem}{Remark}
\newtheorem*{exmp*}{Example}
\newcommand{\N}{\mathbb{N}}

\newcommand{\calC}{\mathcal{C}}

\newcommand{\F}{\mathbb{F}}

\newcommand{\polytime}{\mathrm{P}}
\newcommand{\NP}{\mathrm{NP}}
\newcommand{\ma}{\mathrm{MA}}
\newcommand{\IP}{\mathrm{IP}}
\newcommand{\iP}{\mathrm{ip}}
\newcommand{\Max}{\mathrm{Max}}
\newcommand{\Min}{\mathrm{Min}}
\newcommand{\CP}{\mathrm{CP}}

\newcommand{\NN}{\mathrm{NN}}
\newcommand{\ED}{\mathrm{ED}}
\newcommand{\ov}{\mathrm{ov}}

\newcommand{\Apx}{\mathrm{Apx}}

\newcommand{\rej}{\mathrm{Rej}}
\newcommand{\cc}{\mathrm{cc}}

\newcommand{\dist}{\mathrm{dist}}
\newcommand{\poly}{\mathrm{poly}}
\newcommand{\polylog}{\mathrm{polylog}}

\newcommand{\row}{\mathrm{row}}
\newcommand{\col}{\mathrm{col}}

\newcommand{\spn}{\mathrm{span}}

\global\long\def\then{\Rightarrow}%
\global\long\def\field#1{\mathbb{#1}}%
\global\long\def\ip#1#2{\left\langle #1,#2\right\rangle }%

\title{Finer-grained Reductions \\ in Fine-grained Hardness of Approximation}

\author{
Elie Abboud {\hskip 5em\relax}
Noga Ron-Zewi\thanks{Research of both authors supported in part
    by ISF grant 735/20 and by the European Union (ERC, ECCC, 101076663). Views and opinions expressed are however those of the author(s) only and do not necessarily reflect those of the European Union or the European Research Council. Neither the European Union nor the granting authority can be held responsible for them.} \\ 
Department of Computer Science, University of Haifa \\ \texttt{}\qquad  \texttt{eliabboud1000@gmail.com, noga@cs.haifa.ac.il}
}

\date{}
 
\begin{document}
 
\maketitle

\begin{abstract}

We investigate the relation between $\delta$ and $\epsilon$ required for obtaining a $(1+\delta)$-approximation in time $N^{2-\epsilon}$ for closest pair problems under various distance metrics, and for other related problems in fine-grained complexity. 

Specifically, our main result shows that  if it is impossible to (exactly) solve the  (bichromatic) inner product (IP) problem for vectors of dimension $c \log N$ in time $N^{2-\epsilon}$, then there is no 
$(1+\delta)$-approximation algorithm for (bichromatic) Euclidean Closest Pair running in time $N^{2-2\epsilon}$, where 
$\delta \approx (\epsilon/c)^2$ (where $\approx$ hides $\polylog$ factors). This improves on the prior result due to Chen and Williams (SODA 2019) which gave a smaller polynomial dependence  of $\delta$ on $\epsilon$, on the order of $\delta \approx (\epsilon/c)^6$.
Our result implies in turn that no $(1+\delta)$-approximation algorithm exists for Euclidean closest pair for $\delta \approx \epsilon^4$,  unless an algorithmic improvement for IP is obtained. This in turn is  very close to the  approximation guarantee of $\delta \approx \epsilon^3$ for  Euclidean closest pair, given by the best known algorithm of Almam, Chan, and Williams (FOCS 2016).
By known reductions, a similar result follows for a host of other related problems in fine-grained hardness of approximation. 

Our reduction combines the hardness of approximation framework of Chen and Williams, together with an MA communication protocol for IP over a small alphabet, that is inspired by the MA protocol of Chen (Theory of Computing, 2020).
\end{abstract}

\section{Introduction}\label{sec:intro}

Traditionally, the approach to determine whether a computational problem is tractable was to find out whether
it has a polynomial-time algorithm. Finding such an algorithm implies that the problem is in $\polytime$,
and thus it was considered efficiently computable. Otherwise, if one is interested
in proving that the problem is intractable, we usually lack
the tools to prove lower bounds; instead one relies on hardness assumptions 
which allow us to prove conditional lower-bounds. 
In the classical theory of $\NP$-hardness, the hardness assumption is that $\polytime \neq \NP$, which is known to imply that no polynomial-time algorithm exists for many central computational problems.

In fine-grained complexity, one is interested in pinning down the  \emph{precise} complexity
of \emph{tractable} computational problems. In particular, a central objective in fine-grained complexity is to 
determine the exact exponent  in the time complexity of
 problems already known to be in $P$. More concretely, given a problem with input length $n$ known
to be solvable in $t(n)$-time, is it possible to solve the problem 
in time $t(n)^{1-\epsilon}$  for some $\epsilon>0$?
 This is  motivated by the fact that despite rigorous study of many central
computational problems in $P$, we have failed to improve on the running time of their best-known algorithms (see for example the survey \cite{Wil18} for a list of such problems). This motivates the question of whether there is an inherent difficulty in the problem that prevents us from finding faster algorithms.

Once more, we typically lack the tools to prove lower bounds, and we  thus instead rely on hardness assumptions to obtain conditional lower bounds for problems in $\polytime$.
One popular such conjecture has been the \emph{Strong Exponential Time Hypothesis} (SETH), which postulates that for any $\epsilon>0$, there exists an integer $k=k(\epsilon)$ so that
it is impossible to solve $k$-SAT on $n$ variables  in time  $2^{(1-\epsilon)n}$ \cite{IP01}.

Another popular conjecture is the \emph{Orthogonal Vector Conjecture} (OVC) which in the low-dimensional regime posits that for any $\epsilon>0$, there exists a $c_{\ov}= c_{\ov}(\epsilon)$ such that given a pair of sets $A,B\subseteq\{0,1\}^{d}$ of cardinality $N$ each and of dimension $d = c_{\ov} \cdot \log N$,  it is impossible to determine whether there exists a pair $(a,b)\in A\times B$
satisfying that $\ip a b=0$ in $N^{2-\epsilon}$ time \cite{GIKW19}. It is known that  SETH implies OVC \cite{Wil05}, and so OVC is at least as plausible as SETH. In terms of algorithms, it is known how to solve the OV problem in time $N^{2-\epsilon}$ with $c = \exp(1/\epsilon)$ \cite{AWY15, CW21}, which implies that $c_{\ov} \geq \exp(1/\epsilon)$.

A related assumption is the \emph{inner product} (IP) \emph{assumption} which postulates that for any $\epsilon>0$, there exists a $c_{\iP}= c_{\iP}(\epsilon)$ such that given a pair of sets $A,B\subseteq\{0,1\}^{d}$ of cardinality $N$ each and of dimension $d = c_{\iP} \cdot \log N$, and an integer $\sigma \in \{0,1,\ldots,d\}$,  it is impossible to determine whether there exists a pair $(a,b)\in A\times B$
satisfying that $\ip a b=\sigma$ in $N^{2-\epsilon}$ time. Once more, since the OV problem is a special case of the IP problem, the IP assumption is at least as plausible as OVC. Indeed, the best known algorithms for the IP problem are only able to solve this problem in time $N^{2-\epsilon}$ with $c \approx 1/\epsilon$ \cite{AW15}, and this only imposes that $c_{\iP} \gtrapprox  1/\epsilon$.\footnote{We use $\approx$, $\gtrapprox$, $\lessapprox$ to hide $\polylog$ factors.} 
 
In recent years, there has been a flurry of work showing fine-grained lower bounds for many central computational problems in P, based on the above assumptions.  A main challenge in showing such fine-grained lower bounds based on these assumptions  is that one must carefully design the reductions so that they run fast enough as not to supersede the lower bound assumptions. 

One fundamental problem for which such fine-grained reductions were shown is the
\emph{Closest Pair}  (CP) \emph{problem}. In this problem, given a distance metric $\dist: \{0,1\}^d \times \{0,1\}^d \to \field R^+$, and given a pair of sets $A,B\subseteq\{0,1\}^{d}$, the goal is to find a pair $(a,b)\in A\times B$ which minimizes $\dist(a,b)$. 
This problem was studied for various metrics such as Hamming, $\ell_{p}$, and edit distance, and it has many applications, for example in computational geometry, geographic
information systems \cite{Hen06}, clustering \cite{Zah71,Alp10},
and matching problems \cite{WTFX07}, to name a few.
For concreteness, in what follows we restrict our attention only to the Euclidean $\ell_2$ metric, though many of the results we mention hold also for other metrics. 

In the low-dimensional regime $d = c \log N$, the trivial algorithm for (Euclidean) closest pair runs in time $\tilde O(N^2)$, while for small values of $c$,
the best known algorithm  runs in time $N^{O(c)}$
\cite{Rabin76, KM95} (so a truly sub-quadratic algorithm is only known for small values of $c$ close to zero).
  On the other hand, in \cite{AW15} it was shown that assuming OVC, for any $\epsilon>0$ there exists $c = c(\epsilon)$ so that
no algorithm can solve this problem in time $N^{2-\epsilon}$.

\subsection{Fine-grained hardness of approximation}

Given the above state of affairs, it is natural to ask whether relaxing the requirements and settling  
for an approximate ``close-enough'' answer can help in designing faster fine-grained algorithms. For example, it is known that for the (Euclidean) CP problem, one can obtain a $(1+\delta)$-approximation with running time $N^{2-\epsilon}$ for $\delta \approx \epsilon^3$ (for any dimension $d \leq N^{1-\epsilon}$) \cite{ACW16}, which is much faster than the best-known exact algorithm.

In terms of impossibility results, known fine-grained reductions can typically be adapted to the approximate setting, based on appropriate 
\emph{gap assumptions},  such as Gap-SETH.\footnote{ The Gap-SETH assumption asserts that for any $\epsilon>0$, there are $k$ and $\delta >0$,  so that no $2^{(1-\epsilon)n}$-time algorithm can, given a $k$-CNF on $n$ variables, distinguish between the case that it is satisfiable, and the case that any assignment satisfies at most an $(1-\delta)$-fraction of its clauses.} 
 In the theory of NP-hardness, it is often possible to base hardness of gap-problems on hardness of exact  problems using PCPs. However, a major barrier  in applying this approach in the fine-grained setting (for example for the purpose of reducing  SETH to Gap-SETH) is the large (super-constant) blow-up in the length of existing PCPs, which translates into a large (super-constant) blow-up in the number of variables $n$ in the reduction. 

Nevertheless, in a recent breakthrough, Abboud, Rubinstein, and Williams \cite{ARW17}   have shown how to utilize PCP machinery (specifically, the sumcheck protocol) for showing fine-grained hardness of approximation results based on \emph{non-gap assumptions}. Since then,  
many works have utilized this framework for showing fine-grained hardness of approximation results for many central problems in P, based on non-gap assumptions such as SETH or OVC 
(see the recent surveys \cite{RW19, FSLM20} for a description of this line of work). 

In particular, for the CP problem, Rubinstein \cite{Rub18} has shown that assuming OVC, for any $\epsilon >0$ there exists $\delta = \delta(\epsilon)$ such that there is no 
$(1+\delta)$-approximation algorithm for (Euclidean)  CP running in time $N^{2-\epsilon}$. This rules out truly sub-quadratic approximation algorithms, running, say, in time $f(\delta) \cdot N^{1.99}$.  However, the obtained dependence of $\delta$ on $\epsilon$ is far from optimal, specifically
$\delta =\exp(- c_{\ov}/\epsilon)$, where $c_{\ov}=c_{\ov}(\epsilon) \geq \exp(1/\epsilon)$ is the constant guaranteed by the OVC conjecture.

In a follow-up work, Chen and Williams \cite{CW19} have shown an improved hardness of approximation result for CP in which $\delta$ only depends polynomially on $\epsilon$. Specifically, they showed that if the IP assumption holds, then for any $\epsilon>0$ there is no 
$(1+\delta)$-approximation algorithm for (Euclidean) CP running in time $N^{2-\epsilon}$, where $\delta = \poly(\epsilon/c_{\iP})$ and $c_{\iP} = c_{\iP}(\epsilon) \gtrapprox 1/\epsilon$  is the constant guaranteed by the IP assumption.
However, the obtained dependence of $\delta$ on $\epsilon$ was still quite small, on the order of $\delta \approx (\epsilon/c_{\iP})^6 \lessapprox \epsilon^{12}$. This is still quite far from the dependence obtained by the best known approximation algorithm for CP which gives an $(1+\delta)$-approximation in time $N^{2-\epsilon}$ for $\delta \approx \epsilon^3$.

In this work, we investigate the question of whether the dependence of $\delta$ on $\epsilon$ can even be further improved, potentially to match the best known approximation algorithm.

\subsection{Our results}

Recall that by the discussion above, the best known approximation algorithm for (Euclidean) CP gives an $(1+\delta)$-approximation in time $N^{2-\epsilon}$ for $\delta \approx \epsilon^3$, while the best-known hardness of approximation result shows that if the IP assumption holds, then no $(1+\delta)$-approximation algorithm running in time $N^{2-\epsilon}$ exists for
$\delta \approx  (\epsilon/c_{\iP})^6 \lessapprox \epsilon^{12}$, where $c_{\iP} = c_{\iP}(\epsilon) \gtrapprox 1/\epsilon$  is the constant guaranteed by the IP assumption.
Thus there remains a large polynomial gap between the upper and lower bounds, and our main result narrows this gap.

\begin{thm}\label{thm:main}
Suppose that the IP assumption holds, i.e.,  for any $\epsilon'>0$, there exists a $c_{\iP}= c_{\iP}(\epsilon')$ such that given a pair of sets $A,B\subseteq\{0,1\}^{d}$ of cardinality $N$ each and of dimension $d = c_{ip}(\epsilon') \cdot logN$, and an integer $\sigma \in \{0,1,\ldots,d\}$,  it is impossible to find a pair $(a,b)\in A\times B$
satisfying that $\ip a b=\sigma$ in $N^{2-\epsilon'}$ time. 

Then for any $\epsilon>0$, there is  $\delta = \tilde \Theta(  ( \frac \epsilon {c_{\iP}(\epsilon/2)})^2)$,  so that any algorithm running in time 
$N^{2-\epsilon}$ cannot 
$(1+\delta)$-approximate Euclidean CP.
\end{thm}

Recall that it is known how to solve the IP problem in time $N^{2-\epsilon}$ for dimension $d = c \log N$ with 
$c\approx 1/\epsilon$, and so it must hold that $c_{\iP}  \gtrapprox 1/\epsilon$. If we assume that $c_{\iP} \approx 1/\epsilon$, then the above theorem gives a dependence of $\delta$ on $\epsilon$ of the form $\delta \approx \epsilon^4 $, which is very close to the dependence of $\delta \approx \epsilon^3$ given by the best known algorithm. Moreover, improving the dependence in the above theorem to $\delta \approx \frac \epsilon {c_{\iP}}$ would imply an algorithmic improvement on the IP problem. We leave the question of determining the exact dependence of $\delta$ on $\epsilon$ as an interesting open problem for future research.

By known reductions, the above theorem gives a similar improvement for a host of other problems in fine-grained hardness of approximation such as closest pair with respect to other metrics such as Hamming,  $\ell_p$-norm for any constant $p>0$, and edit distance, Furthest Pair and approximate nearest neighbor in these metrics, and additive approximations to Max-IP and Min-IP, see Section \ref{sec:apps} for more details. 

Finally, we remark that the above theorem also holds under OVC (or SETH), but is less meaningful, since as discussed above, for the OV problem we have that $c_{\ov} \geq \exp(1/\epsilon)$. 

\subsection{Proof overview}

Next we give an overview of our proof method, and how it improves on prior work. To this end, we first describe the general framework presented in \cite{ARW17} for obtaining fine-grained hardness of approximation results based on MA communication protocols. Then we discuss the work of Rubinstein \cite{Rub18} who relied on this framework to give the first fine-grained hardness of approximation result for CP, albeit with an exponential dependence of $\delta$ on $\epsilon$, and the work of Chen and Williams \cite{CW19} who improved this dependence to polynomial. Following this, we turn to discuss our proof method that obtains a tighter polynomial relation.

\paragraph{Fine-grained hardness of approximation via MA communication \cite{ARW17}:}  In a Merlin-Arthur (MA) communication protocol for a function $f: \{0,1\}^d \times \{0,1\}^d \to \{0,1\}$, two players
Alice and Bob wish to compute
$f(a,b)$, where Alice is given as input only
$a \in\{0,1\}^d$, and Bob is given as input only $b \in \{0,1\}^d$.
To this end, Alice and Bob engage in a randomized (public coin) communication protocol, where their goal is to use as little communication as possible. 
To aid them with this task, there is also  a (potentially malicious)
prover Merlin who sees Alice's and Bob's inputs, and before any communication
begins Merlin sends Alice a short message $m$, which can be thought of as a ``proof'' or ``advice''. 
The requirement is that if $f(a,b)=1$, then there must exist
some message $m$ from Merlin on which Alice accepts with probability
$1$. Otherwise, if $f(a,b)=0$, then for any possible message $\tilde m$ from Merlin, Alice accepts with probability
at most $\frac 12$ on $\tilde m$.\footnote{The rejection probability can be increased by executing the communication phase between Alice and Bob independently for multiple times and rejecting if and only if all invocations reject.} 

In \cite{ARW17}, it was shown that an efficient MA communication protocol for \emph{set disjointness}\footnote{Recall that \emph{set disjointness} is the function $\mathrm{disj}:\{0,1\}^d \times \{0,1\}^d \to \{0,1\}$ which satisfies that $
\mathrm{disj}(a,b)=1$ if and only if the supports of $a$ and $b$ are disjoint, i.e., if $\ip a b =0$.} implies a fine-grained reduction from OV to an approximate version of Max-IP in which given two sets $A, B\subseteq\{0,1\}^d$ of cardinality $N$ each, the goal is to output a number sufficiently close to $M:=\max_{a \in A, b \in B} \ip a b$.

To see how a reduction as above can be constructed, suppose that there exists an MA communication protocol for set disjointness with Merlin's message length $L$, communication complexity $\cc$ between Alice and Bob, and randomness complexity $R$. Suppose furthermore that we are given an instance $A,B \subseteq \{0,1\}^d$ of OV, where $|A|=|B|=N$. Then for each possible Merlin's message $m \in \{0,1\}^L$, we construct an instance $A_m,B_m \subseteq \{0,1\}^{ 2^{\cc+R} }$ of Max-IP, where $|A_m| = |B_m|=N$.

Fix $m \in \{0,1\}^L$. Then the set $A_m$ is obtained from $A$ by mapping each element $a \in A$  to a binary vector $a_m$ that contains an entry for each possible transcript $\Gamma \in \{0,1\}^{\cc}$ and randomness string $r \in \{0,1\}^R$ (so $a_m$ has length $2^{\cc+R}$), and whose $(\Gamma, r)$-entry equals $1$ if and only if $\Gamma$ is consistent with $r$ and $a$, and Alice accepts on input $a$, randomness string $r$, and transcript $\Gamma$.
The set $B_m$ is obtained analogously from $B$. 

Then the main observation is that for some $m \in \{0,1\}^L$, we have that the $(\Gamma,r)$-entry of both $a_m$ and $b_m$ equals $1$ if and only if Alice accepts on Merlin's message $m$, inputs $a$ and $b$, and randomness string $r$. Consequently, if there exists $(a,b) \in A \times B$ so that $\ip a b =0$ (i.e., $f(a,b)=1$), then there exists a Merlin's message $m$ on which Alice accepts with probability $1$ on inputs $a$ and $b$, and consequently we have that the corresponding vectors
$(a_m,b_m) \in A_m \times B_m$ satisfy that $\ip {a_m} {b_m} = 2^R$. 
On the other hand, if 
$\ip a b \neq 0$ (i.e., $f(a,b) =0$) for any $(a,b) \in A \times B$, then for any Merlin's message $\tilde m$, and on any inputs $(a,b) \in A \times B$, Alice accepts with probability at most $\frac 1 2$, and so $\ip {a_{\tilde m}} {b_{\tilde m}} \leq \frac 1 2 \cdot 2^R$ for any pair
$(a_{\tilde m},b_{\tilde m}) \in A_{\tilde m} \times B_{\tilde m}$. This gives the desired gap, showing that OV reduces to $2^L$ instances of approximate Max-IP. 

To obtain a fine-grained reduction, one must make sure that $\cc$, $R$ and $L$ are not too large,
so that the total construction time of the reduction is at most $N^{\epsilon}$. To achieve this, one can use the MA communication protocol of Aaronson and Wigderson \cite{AW09} for set disjointness in which all these quantities are upper bounded by $\approx \sqrt{d}$. 

For $d = c \log N$, this gives that $2^\cc$, $2^R$, and $2^L$ are all upper bounded by $2^{\tilde O(\sqrt{ \log N})} \ll N^{\epsilon}$, and so the reduction can be constructed in time  $N^{\epsilon}$. 
Next we describe the MA communication protocol of \cite{AW09}, as hardness of approximation results for CP (including ours) crucially rely on its properties.

\paragraph{MA communication protocol for set disjointness \cite{AW09}:} The MA communication protocol for set disjointness of Aaronson and Wigderson \cite{AW09} relies on the influential sumcheck protocol of \cite{LFKN92}, and it proceeds as follows.

Let $a \in\{0,1\}^{d}$ be Alice's input. 
Slightly abusing notation, we view $a$ as a $\sqrt{d} \times \sqrt{d}$ binary matrix in the natural way, and we let $\hat a$ denote the $p \times \sqrt{d}$ matrix obtained by encoding each column of $a$ with a systematic Reed-Solomon code $\mathrm{RS}_{\sqrt{d}, p} : \F_p^{\sqrt{d}} \to \F_p^{p}$ of degree $\sqrt{d}$ over a prime field of size $p \approx 4\sqrt{d}$.\footnote{The systematic Reed-Solomon code $\mathrm{RS}_{d,p}: \F_p^d \to \F_p^p$ of degree $d$ over a prime field of size $p>d$ is a linear map, defined as follows. To encode a message $m = (m(0),\ldots, m(d-1))  \in \F_p^d$, one finds the (unique) degree $d-1$ polynomial $P_m(X) \in \F_p[X]$ which satisfies that $P_m(i) = m(i)$ for any $i = 0, \ldots , d-1$, and lets $\mathrm{RS}_{d,p}(m)= (P_m(0), \ldots, P_m(p-1))$. The code is called \emph{systematic} since the message is a prefix of its encodings.
The code has distance at least $p-d+1$ since any pair of distinct degree $d-1$ polynomials can agree on at most $d-1$ points. }
 Let $\hat b$ be defined analogously.

In the protocol, Merlin first computes the pointwise product $\hat a \star \hat b \in \F_p^{p\times \sqrt{d}}$, and then sends Alice the sum $m \in \F_p^p$ of the columns of $\hat a \star \hat b$ (where arithemtic is performed mod $p$). Alice first checks that $m$ is a codeword of $\mathrm{RS}_{2\sqrt{d},p}$, and that the first $\sqrt{d}$ entries of $m$ are all zero, otherwise she rejects and aborts.
Then Alice and Bob jointly sample a random index $i \in [p]$, Bob sends Alice the $i$'th row of $\hat b$, Alice computes its inner product with the $i$'th row of $\hat a$, and accepts if and only if this product equals $m(i)$ (where once more, arithmetic is performed mod $p$).
 
 To see that the protocol is complete, note first that if $a$ and $b$ are disjoint, then $a \star b$ is the all-zero matrix. Consequently, by the systematic property of the Reed-Solomon encoding, the first $\sqrt{d}$ rows of $\hat a \star \hat b$ are also identically zero, which implies in turn that the first $\sqrt{d}$ entries of $m$ are identically zero. Furthermore, since the product of
two polynomials of degree at most $\sqrt{d}$ is a polynomial of degree at most $2 \sqrt{d}$, it follows that $m$ is a codeword of $\mathrm{RS}_{2\sqrt{d},p}$. Thus, both Alice's checks will clearly pass. It can also be verified that by construction, the inner product of the $i$'th rows of $\hat a$ and $\hat b$ equals $m(i)$, and so Alice accepts with probability $1$.

To show soundness, suppose that $a$ and $b$ intersect, and let $\tilde m$ denote Merlin's message. We may assume that $\tilde m$ is a codeword of $\mathrm{RS}_{2\sqrt{d},p}$, and that the first $\sqrt{d}$ entries of $\tilde m$ are all zero, since otherwise Alice clearly rejects. But on the other hand, since $a$ and $b$ intersect, then $a \star b$ has a $1$-entry, say in the $j$-th row, and since $p>\sqrt{d}$, the sum of entries in the $j$'th row of $a \star b$ is non-zero mod $p$, which implies in turn that $m(j) \neq 0$. Thus, we conclude that $\tilde m$ and $m$ are distinct codewords of $\mathrm{RS}_{2\sqrt{d},p}$ -- a code of distance at least $\frac p 2$ -- and so they must differ by at least $\frac 1 2$ of their entries. But this implies in turn that with probability at least $\frac 1 2$ over the choice of $i$, it holds that 
$\tilde m(i)\neq m(i)$, in which case the inner product of the $i$'th row of $\hat a$ and $\hat b$ will be different than $\tilde m(i)$, which will cause Alice to reject.

Finally, it can also be verified that in this protocol, $\cc$, $R$, and $L$ are all upper bounded by $\approx \sqrt{d}$.

\paragraph{Hardness of approximation for CP with exponential dependence \cite{Rub18}:}

In \cite{Rub18}, Rubinstein utilized the above framework to show fine-grained hardness of approximation for CP. 
The starting point of \cite{Rub18}  is a simple linear-time reduction
from $\delta$-additive approximation for Max-IP\footnote{In a $\delta$-\emph{additive approximation} for Max-IP, given $A, B\subseteq\{0,1\}^d$ of cardinality $N$ each, the goal is to output a number in $[M - \delta \cdot d, M]$, where $M:=\max_{a \in A, b \in B} \ip a b$.}  to an $(1+\Theta(\delta))$-approximation for (Euclidean) CP. Thus, to show that no algorithm can find an $(1+\Theta(\delta))$-approximation for (Euclidean) CP in time $N^{2-\epsilon}$, it suffices to show that no algorithm can find a $\delta$-additive approximation for Max-IP in time $N^{2-\epsilon}$. 

The \cite{ARW17} framework discussed above generates instances of Max-IP of dimension $2^{\cc+R}$ and additive gap  of $\frac 1 2\cdot 2^R$, which gives $\delta:=\Theta(2^{-\cc})$. However, the MA protocol of \cite{AW09} described above only gives $\cc \approx \sqrt{d} $ which is super-constant for a super-constant dimension $d$, and consequently only yields a \emph{sub-constant} $\delta$.

To deal with this, \cite{Rub18} first utilized the fact (previously utilized also in \cite{ARW17}) that the \cite{AW09} protocol described above works equally well on skewed matrices of dimensions $\frac d T \times T$, in which case we have that $L=\frac d T \cdot \log p$ and $\cc= T \cdot \log p$. 
Thus, assuming $d = c \log N$, to achieve $2^L \leq N^{\epsilon}$, one can set $T = \frac c \epsilon \cdot \log p $, which gives in turn $\cc=\frac c \epsilon \cdot \log ^2p$.

However, this is still not quite enough since the MA protocol of \cite{AW09} requires setting $p >  \sqrt{d}$ because of the use of Reed-Solomon codes that are only defined over a large alphabet, and consequently the communication complexity is still super-constant. However, the main observation in \cite{AW09} is that the protocol can actually be executed using any error-correcting code with a \emph{multiplication property}\footnote{Informally, we say that a linear code $C:\F^k \to \F^n$ has a \emph{multiplication property} if the set $\spn\{C(m) \star C(m') \mid m,m' \in \F^k  \}$ has sufficiently large distance.}. 
Relying on this observation, Rubinstein replaced the Reed-Solomon codes in the protocol of \cite{AW09} with algebraic-geometric (AG) codes that satisfy the multiplication property over a constant-size alphabet.
This reduced the communication complexity to $ \approx c/\epsilon$, yielding in turn an approximation factor of $\delta = 2^{- \tilde \Theta (c/\epsilon)}$. 

\paragraph{Polynomial dependence \cite{CW19}:}

While \cite{Rub18} gave the first non-trivial hardness of approximation result for CP, a downside of this result was that the approximation factor $\delta$ depended exponentially on the running time parameter $\epsilon$. In the follow-up work \cite{CW19}, Chen and Williams showed how to reduce this dependence to just polynomial.

The main observation of Chen and Williams was that instead of thinking of the output of the MA protocol of  \cite{Rub18} as being just accept or reject, one can view the output as being \emph{short vectors} $a',b' \in \F_p^T$ (namely, the $i$'th row of $\hat a, \hat b$, respectively), and $\sigma' \in \{0,1,\ldots, d\}$ (namely, the $i$'th entry of Merlin's message $m$), where $a'$ only depends on Alice's input $a$ and the randomness string, $b'$ only depends on Bob's input $b$ and the randomness string, and $\sigma'$ only depends on Merlin's message and the randomness string. The requirement then is that if $\ip a b =0$ for some $(a,b) \in A \times B$,
then for some Merlin's message $m$,
$\ip {a'} {b'} = \sigma'$ with probability $1$, while if $\ip a b \neq 0$ for any $(a,b) \in A \times B$,
then for any Merlin's message $\tilde m$,
then $\ip {a'} {b'} \neq \sigma'$ with probability at least $\frac 1 2$.

Chen and Williams then suggested to create an instance $A_m, B_m$  for any Merlin's message $m$, where the set $A_m$  is obtained from $A$ by simply mapping each element $a \in A$ to a vector $a_m \in \F_p^{T \cdot 2^R}$ that is the concatenation of all possible output vectors $a'$ for all possible randomness strings, and analogously for $B_m$. The advantage is that now 
the dimension of the vectors in $A_m$ and $B_m$ is much shorter than in  
\cite{Rub18}. However,  a disadvantage is that now the alphabet is not binary anymore, and an even more serious problem is that the soundness guarantee is only that $\ip {a'} {b'} \neq \sigma'$, so the reduction does not seem to produce any gap.

To deal with these issues, Chen and Williams use an \emph{encoding lemma} which gives mappings $g,h$ and a value $\Gamma$, where $g,h$, and $\Gamma$ only depend on $p$ and $T$, so that $g(a',\sigma')$ and $h(b',\sigma')$ are binary vectors of length $\poly(p, T)$ satisfying that if $\ip {a'} {b'} = \sigma'$ then $\ip {g(a',\sigma')} {h(b',\sigma')} =\Gamma$, while if $\ip {a'} {b'} \neq \sigma'$ then $\ip {g(a',\sigma')} {h(b',\sigma')} < \Gamma$ (see Lemma \ref{lem:enc} for a formal statement). This produces the desired additive gap, on the order of $\Omega(2^R)$. Since the encoding lemma increases the dimension of the vectors only by a factor of $\poly(p,T)$, this yields an approximation factor of $\delta =\frac{1} {\poly(p,T)}=  \poly(\frac \epsilon c)$.

We note that a delicate issue that should be dealt with in the reduction is that the encoding lemma works over the integers, while the protocol works over finite fields, and in particular, over non-prime fields, as AG codes are only known to exist over non-prime fields. Additionally, Chen and Williams show that the reduction works equally well when using the IP problem instead of OV as its starting point, and using an MA communication protocol for IP similar to that of \cite{AW09}. This can potentially allow for a smaller value of $c$ as it is only known how to solve (exact) IP in time 
 $N^{2-\epsilon}$ up to a dimension of $c \log N$ for $c \approx 1/\epsilon$.

\paragraph{This work - tighter polynomial dependence:} While \cite{CW19} obtained a polynomial dependence of $\delta$ on $\epsilon$, the dependence was quite small, on the order of $\delta \approx  (\frac \epsilon c)^6$, and in the current work we show how to improve the dependence to $\delta \approx (\frac  \epsilon c)^2$.

To this end, we first observe that one reason for the small polynomial dependence obtained in \cite{CW19} was the large polynomial dependence of the dimension of the resulting vectors in the encoding lemma on the alphabet size $p$. While in the protocol of \cite{Rub18} the field size $p$ can be made constant using AG codes, the field still needs to be of characteristic at least $T$, since otherwise the sum of entries in a non-zero row of $a \star b$ may sum to zero over $\F_p$, and consequently the soundness analysis will not go through.

To reduce the alphabet size, we first design a new MA protocol in which the alphabet size is only \emph{polylogarithmic} in $T$ (see Theorem \ref{thm:ma}).
This protocol is inspired by the MA protocol of \cite{Che18} for IP which achieved communication complexity $O(\sqrt{d \log d\log \log d})$, improving on the communication complexity of $O(\sqrt{d} \log d)$ of \cite{AW09}.
In a nutshell, Chen's idea was to execute the original MA protocol of  \cite{AW09} multiple times over different small prime fields, hoping that if $\ip a b \neq 0$, then $\ip a b$ is also non-zero modulo many of the primes, and so the protocol will be executed correctly.
Chen showed that this is indeed possible to achieve using $O(\log d)$ distinct primes of cardinality at most $\polylog(d)$ each.

We observe that for skewed matrices of  dimensions $\frac d T \times T$, it in fact suffices to execute the protocol with $O(\log T)$ distinct primes of cardinality at most $\polylog(T)$ each. While this choice does not necessarily guarantee the property above that if $\ip a b \neq 0$, then $\ip a b$ is also non-zero modulo many of the primes, this turns out to still suffice for a correct execution of the protocol. 

We then further observe that such a protocol can be used in the framework of \cite{CW19} to obtain an improved hardness of approximation result for Max-IP. Once more, a delicate issue is how to use the encoding lemma in the presence of many different non-prime fields.   

To the best of our knowledge, this is the first use of the techniques underlying the improved MA protocol of \cite{Che18} for showing a fine-grained hardness of approximation result.

\paragraph{Paper organization.} The rest of the paper is organized as follows. We begin in Section 2 below with the required notation and terminology with respect to fine-grained complexity problems and error-correcting codes. Then in Section \ref{sec:ma} we present our improved MA protocol over a small alphabet, while in Section
\ref{sec:ip_to_apx_max_ip} we show how to use this protocol for obtaining an improved reduction from IP to approximate Max-IP. Finally, in Section \ref{sec:apps} we show implications of our latter result to showing hardness of approximation results for closest pair, as well as other related problems in fine-grained complexity.

\section{Preliminaries}\label{sec:prelim}

We start by setting some general notation.
For a positive integer $d$, we let $[d]:=\{1,2, \ldots, d\}$.  For convenience, we often view a vector $a \in \Sigma^d$ as a function $a:[d] \to \Sigma$, and we let $a(i)$ denote the $i$-th entry of $a$. 
For a pair of vectors $a,b \in \N^d$, we let $\ip a b := \sum_{i=1}^d a(i) \cdot b(i)$ denote their inner product,  and we let $a \star b \in \N^d$ denote their pointwise product, given by $(a \star b)(i) = a(i) \cdot b(i)$ for $i \in [d]$. For $a,b \in \Sigma^d$, we let $\Delta(a,b):=|\{i \in [d] \mid a(i) \neq b(i)\}|$ denote their \textsf{Hamming distance}.
For an $n \times k$ matrix $A$ and $i\in [n]$ ($j \in [k]$, respectively),
we let $\row_i(A)$ ($\col_j(A)$, respectively) denote the $i$-th row ($j$-th column, respectively) of $A$.

\subsection{Problems in fine-grained complexity}

Below we list the main fine-grained problems that we will be concerned with in this paper.

\begin{defn}[Inner Product ($\IP$)] In  the \textsf{ inner product $\IP_{N,d}$ problem},
given two sets $A, B\subseteq\{0,1\}^d$
of cardinality $N$ each, and an integer $\sigma \in \{0,1,\ldots, d\}$, the goal is to determine whether there exists a pair $(a,b)\in A\times B$ satisfying that $\ip a b = \sigma$.
\end{defn}

\begin{defn}[Maximum  Inner Product ($\Max$-$\IP$)] In  the \textsf{maximum inner product $\Max$-$\IP_{N,d}$ problem},
given two sets $A, B\subseteq\{0,1\}^d$
of cardinality $N$ each, the goal is to compute 
$M:=\max_{a \in A, b \in B} \ip a b.$
\end{defn}

For the approximate version of $\Max$-$\IP$, defined next, we will consider the less standard \emph{additive} approximation version that will be useful for obtaining  hardness of approximation  for the closest pair problem.

\begin{defn}[Approximate Maximum Inner Product ($\Apx$-$\Max$-$\IP$)] 
Let $\delta >0$ be a parameter.
In  the \textsf{(additive) approximate maximum inner product $\delta$-$\Apx$-$\Max$-$\IP_{N,d}$ problem},
given two sets $A, B\subseteq\{0,1\}^d$ of cardinality $N$ each, the goal is to output a number in $[M-\delta\cdot d, M]$, where $M:=\max_{a \in A, b \in B} \ip a b$.
\end{defn}

\begin{defn}[Closest Pair ($\CP$)] 
Let $\dist:\{0,1\}^d \times\{0,1\}^d \to\field R^+$ be a distance function. In the \textsf{closest pair  $\CP_{N,d,\dist}$ problem},
given two sets $A, B\subseteq\{0,1\}^d$
of cardinality $N$ each, the goal is to compute 
$M:=\min_{a \in A, b \in B} \dist(a,b).$
\end{defn}

\begin{defn}[Approximate Closest Pair ($\Apx$-$\CP$)] 
Let $\dist:\{0,1\}^d \times\{0,1\}^d \to\field R^+$ be a distance function, and let $\delta>0$ be a parameter. In the \textsf{approximate closest pair  $\delta$-$\Apx$-$\CP_{N,d,\dist}$ problem},
given two sets $A, B\subseteq\{0,1\}^d$
of cardinality $N$ each, the goal is to output a number
in $[M, (1+\delta)M]$, where 
$M:=\min_{a \in A, b \in B} \dist(a,b).$
\end{defn}

\subsection{Error-correcting codes}

Our reduction will make use of error-correcting codes. In what follows, we first present some general notation and terminology with respect to error-correcting codes, and then describe the kind of codes we shall use for our reduction.

Let $\Sigma$ be a finite alphabet, and $k, n$ be positive integers (the \textsf{message length} and the
\textsf{codeword length}, respectively). An \textsf{(error-correcting) code} is an injective map $C: \Sigma^k \to \Sigma^n$. The elements in the domain of $C$ are called \textsf{messages}, and the elements in the image of $C$ are called \textsf{codewords}.
We say that $C$ is \textsf{systematic} if the message is a prefix of
the corresponding codeword, i.e., for every $x \in \Sigma^k$ there
exists $z \in \Sigma^{n-k}$ such that $C(x) = (x,z)$.
The \textsf{rate} of a code $C: \Sigma^k \to \Sigma^n$ is the ratio
$\rho := \frac{k} {n}$.  The \textsf{relative distance} $\dist(C)$ of
$C$ is the maximum $\delta >0$ such that for every pair of distinct
messages $x,y\in \Sigma^k$ it holds that $\Delta(C(x),C(y))\ge\delta$.

If $\Sigma=\F$ for some finite field $\F$, and $C$ is a linear map
between the vector spaces $\F^k$ and $\F^n$ then we say that $C$ is
\textsf{linear}. The \textsf{generating matrix} of a linear code
$C: \F^k \to \F^n$ is a matrix $G \in \F^{n \times k}$ such that
$C(x) = G\cdot x$ for any $x\in \F^k$. We say that a linear code $C$ is \textsf{explicit} if 
$G$ can be generated in time $\poly(n)$. 

For our reduction, we shall require linear codes over a small (constant-size, independent of the codeword length) alphabet, satisfying the \emph{multiplication property}, which informally says that all pointwise products of pairs of codewords span a code of large distance. 
Such codes can be obtained from the AG codes of \cite{SAKSD01} (see also \cite[Theorem 2.4]{Rub18}).

\begin{thm}[\cite{SAKSD01}; \cite{Rub18}, Theorem 2.4]\label{thm:ag}
There exists a constant integer $p_0$ so that for any prime $p \geq p_0$, there exist two explicit code families $\calC=\{C_k\}_{k \in \N}$ and $\calC_{\star}= \{(C_{\star})_k\}_{k \in \N}$ so that the following hold for any $k \in \N$:
\begin{itemize}
\item $C_k, (C_\star)_k$ are  systematic linear codes over $\F_{p^2}$ of  relative distance at least $0.1$ and rate at least $0.1$.
\item $C_k$ has message length $k$.
\item For any $x,y\in (\F_{p^2})^k$, $C_k(x) \star  C_k(y)$ is a codeword of $(C_\star)_k$. 
\end{itemize}
\end{thm}

\section{MA protocol for IP over a small alphabet}\label{sec:ma}

In this section, we will provide an MA protocol for IP over a small alphabet. The protocol will later be used in Section \ref{sec:ip_to_apx_max_ip} below to show a reduction from IP to Apx-Max-IP. A  protocol with similar guarantees was implicitly given in \cite{CW19}, albeit with an exponentially larger alphabet of  $q=\Theta(T)$. 

\begin{thm}[MA Protocol for $\IP$ over a small alphabet]\label{thm:ma} 
For any sufficiently large integer $T$, there is an integer  $q=O(\log ^2 T)$, so that for 
 any integer $d$ which is a multiple of $T$ there is an $\ma$ Protocol which satisfies the following:
\begin{enumerate}
\item Alice is given as input a vector $a \in\{0,1\}^{d}$  and an integer $\sigma \in \{0,1,\ldots,d\}$,
Bob is given as input a vector $b \in \{0,1\}^{d}$, and Merlin is given as input
$a$, $b$, and $\sigma$.
\item Merlin sends Alice a message $m$ of (bit) length $L=O(\frac d T \cdot \log^2 T)$.
 Alice reads Merlin's message, and based on this message and $\sigma$, decides whether to reject and abort, or continue.
\item Alice and Bob sample a joint random string $r$ of (bit) length $R=\log(\frac{d}{T})+ \log \log T+O(1)$.
\item Alice outputs  a string
$a' \in \{0,1,\ldots,q\}^T$ and an integer $\sigma' \in \{0,1,\ldots, T \cdot q^2\}$, where $a'$ only depends on Alice's input $a$ and the randomness string $r$,  
and $\sigma'$ only depends on Merlin's message $m$ and $r$,  and Bob outputs a string $b' \in \{0,1,\ldots,q\}^T$, which only depends on Bob's input $b$ and  $r$, so that the following hold:
\begin{itemize}
\item (Completeness) If $\ip a b= \sigma$, then on Merlin's message $m$, Alice and Bob output $a', b'$, and $\sigma'$ so that $\ip {a'} {b'} = \sigma'$ with probability $1$. 
\item (Soundness) If $\ip a b \neq \sigma$, then for any Merlin's message $\tilde m$, Alice and Bob output $a', b',$ and $\sigma'$ so that $\ip {a'} {b'} = \sigma'$ with probability at most $0.98$.
\end{itemize}

Moreover, the running time of both Alice and Bob is $\poly(d)$.
\end{enumerate}
\end{thm}

The main difference between the above protocol and that of \cite{CW19}, is that instead of working over a field of characteristic $\Theta(T)$, we perform the protocol of \cite{CW19} simultaneously over $O(\log T)$ different fields of size $O(\log^2T)$ each.

To this end, we start by fixing some notation.
Let $t$ be an integer such that $t^t=T$. By Lemma 2.4 in \cite{Che18}, for a large enough integer $t$, there exist $10 t$ distinct primes $p_{1} < p_2<\cdots < p_{10 t} \in [t,t^2]$.
Let $q:=t^2$, and note that $t = O(\log T)$ and $q = O(\log^2T)$
 For $\ell \in [10 t]$, let  $C^{(\ell)}, C_\star^{(\ell)}$  be the systematic linear codes over $\F_{p_\ell^2}$ guaranteed by Theorem \ref{thm:ag}, where $C^{(\ell)}$ has message length $\frac d T$ and codeword length $n_\ell:=O(\frac d T)$. 
 Finally, recall that the elements of $\F_{p_\ell^2}$ can be viewed as degree $1$ polynomials over $\F_{p_\ell}$, where multiplication is performed modulo an irreducible polynomial $Q_\ell$ of degree $2$ over $\F_{p_\ell}$. 
 
 Let $a \in\{0,1\}^{d}$ be Alice's input. Slightly abusing notation, we view $a$ as a $\frac d T \times T$ binary matrix in the natural way.
 For $\ell \in [10 t]$, let $a^{(\ell)}$ denote the  $ n_\ell \times T$ matrix over $\F_{p_\ell^2}$ obtained 
by encoding each column of 
$a$ with the code $C^{(\ell)}$. 
View each entry of $a^{(\ell)}$ as a degree $1$ polynomial over $\F_{p_\ell}$, and let $a^{(\ell,0)}, a^{(\ell,1)}$ denote the 
$n_\ell \times T$ matrices over $\F_{p_\ell}$, obtained from $a^{(\ell)}$ by keeping in each entry only the free coefficient and linear coefficient, respectively.
Let $b,b^{(\ell)}, b^{(\ell,0)}, b^{(\ell,1)}$ be defined analogously for $\ell \in [10 t]$. 
In what follows, all arithmetic operations are performed over the reals, unless otherwise stated.

\paragraph{The protocol:}

The protocol proceeds as follows:

\begin{enumerate} 
\item \label{step:merlin}

\begin{enumerate}

\item \label{stp:merlin_1}

Merlin sends
$$m_0 := \sum_{j=1}^T  \col_j(a) \star \col_j(b) \in \{0,1,\ldots,T\}^{d/T}.$$

 \item For $\ell=1, \ldots, 10 t$ 
and $\alpha, \beta \in \{0,1\}$, Merlin sends
 $$m_{\ell,\alpha,\beta}:=\sum_{j=1}^T  \col_j(a^{(\ell,\alpha)}) \star \col_j(b^{(\ell,\beta)}) \in \{ 0,1, \ldots,T\cdot q^2\}^{n_\ell }.$$
 
 \end{enumerate}
 
 \item \label{step:alice_check}

\begin{enumerate}

\item \label{step:alice_check_sum} Alice checks that $\sum_{i=1}^{d/T} m_0(i) = \sigma$.

\item   \label{step:alice_check_con} Alice checks that $m_{\ell,0,0}(i) = m_0(i)$ and  $m_{\ell, 0,1} (i) =m_{\ell,1,0} (i) = m_{\ell,1,1} (i)=0$ for $\ell=1, \ldots, 10 t$ and  $i = 1, \ldots, \frac d T$ .

\item \label{step:alice_check_code}
For $\ell = 1, \ldots, 10 t$, let $m_\ell \in (\F_{p_\ell^2})^{n_\ell}$ given by 
$$  m_\ell = 
      m_{\ell,0,0}   + ( m_{\ell,0,1}   + m_{\ell,1,0} ) \cdot X +  m_{\ell,1,1}   \cdot X^2   \;\;(\mathrm{mod} \;Q_\ell).$$
      Alice checks that $m_\ell$ is a codeword of $C_\star^{(\ell)}$ for $\ell=1, \ldots, 10 t$.

 \end{enumerate}
 
 If any of the checks is unsatisfied, then Alice rejects and aborts.

\item \label{step:random} Alice and Bob jointly sample $\ell_* \in [10 t]$, $i_* \in[n_{\ell_*}]$, and $\alpha_*, \beta_* \in \{0,1\}$.

\item  \label{step:output}
Alice outputs $a' := \row_{i_*}(a^{(\ell_*,\alpha_*)}) \in \{0,1,\ldots,q\}^T$ and $\sigma' := m_{\ell_*,\alpha_*,\beta_*} (i_*)\in \{0,1,\ldots, T \cdot q^2\}$, and Bob outputs  $ b':=\row_{i_*}(b^{(\ell_*,\beta_*)}) \in \{0,1,\ldots,q\}^T$.
\end{enumerate}

It can be verified that the protocol has the required structure, and that the running times of Alice and Bob are as claimed. Next we show completeness and soundness.

\paragraph{Completeness:}

Suppose that $\ip a b = \sigma$, we shall show that in this case Alice and Bob output $a'$, $b'$, and $\sigma'$ so that $\ip {a'} {b'} = \sigma'$ with probability $1$. 

We first show that in this case all of Alice's checks on Step \ref{step:alice_check} always pass.

To this end, first note that  by assumption that $\ip a b = \sigma$, we have that
\begin{equation}\label{eq:complete_1}
\sum_{i=1}^{d/T} m_0(i) = \sum_{i=1}^{d/T} \sum_{j=1}^T  a(i,j) \cdot b(i,j) = \ip a b = \sigma,
\end{equation} 
so Alice's check on Step \ref{step:alice_check_sum} will pass.

We now show that Alice's check on Step \ref{step:alice_check_con} passes. Fix $\ell \in [10 t]$, and recall that $a^{(\ell)}$ is obtained by encoding each column of the matrix $a \in \{0,1\}^{\frac d T \times T}$ with 
a systematic linear code.
Consequently, $a$ is the restriction of
$a^{(\ell)}$  to the first $\frac d T$ rows, and similarly for $b$.
This implies in turn that  for any $i \in [\frac d T]$, we have that
\begin{equation}\label{eq:complete_2}
m_{\ell,0,0}(i) = \ip{\row_i(a^{(\ell,0)})}  {\row_i(b^{(\ell,0)})} = \ip{\row_i(a)}  {\row_i(b)} = m_0(i),
\end{equation}
and 
\begin{equation}\label{eq:complete_3}
m_{\ell,0,1} (i) = m_{\ell,1,0} (i) = m_{\ell,1,1} (i) =0. 
\end{equation}
So Alice's check on Step \ref{step:alice_check_con} will pass as well.

Finally, we show that Alice's check on Step \ref{step:alice_check_code} passes. Fix $\ell \in [10 t]$, and note that
\begin{eqnarray}\label{eq:complete_4}
\nonumber
m_\ell & =&  m_{\ell,0,0}   + ( m_{\ell,0,1}   + m_{\ell,1,0} ) \cdot X +  m_{\ell,1,1}   \cdot X^2   \;\;(\mathrm{mod} \;Q_\ell) \\ \nonumber
& = & \sum_{j=1}^T \bigg[ \col_{j}(a^{(\ell,0)}) \star \col_j(b^{(\ell,0)}) \\ \nonumber
&&  +  \left( col_{j}(a^{(\ell,0)}) \star \col_j(b^{(\ell,1)}) +col_{j}(a^{(\ell,1)}) \star \col_j(b^{(\ell,0)})  \right) \cdot X   \\  \nonumber
&& +  col_{j}(a^{(\ell,1)}) \star \col_j(b^{(\ell,1)})    \cdot X^2 \bigg]  \;\;(\mathrm{mod} \;Q_\ell)
 \\ \nonumber
&= &\sum_{j=1}^T ( \col_{j} (a^{(\ell,0)}) + \col_{j} (a^{(\ell,1)})\cdot X)\star ( \col_{j} (b^{(\ell,0)}) + \col_{j} (b^{(\ell,1)})\cdot X) \;\;\;(\mathrm{mod} \;Q_\ell)\\
&= &\sum_{j=1}^T \col_{j}(a^{(\ell)}) \star \col_j(b^{(\ell)})  \;\;\;(\mathrm{mod} \;Q_\ell).
\end{eqnarray}
Now, since each column of $a^{(\ell)}$ and $b^{(\ell)}$ is a codeword of $C^{(\ell)}$, we have that $\col_{j}(a^{(\ell)}) \star \col_j(b^{(\ell)}) \;\;(\mathrm{mod} \;Q_\ell)$ is a codeword of $C_\star^{(\ell)}$ for any $j \in [T]$. By linearity 
of $C_\star^{(\ell)}$, this implies in turn that $m_\ell$ is a codeword of $C_\star^{(\ell)}$, 
and so Alice's check on Step \ref{step:alice_check_code} will also pass.

Thus, we conclude that all of Alice's checks on Step \ref{step:alice_check} pass.
Furthermore, we clearly have that 
$$\langle  a', b' \rangle  =\ip {\row_{i_*}(a^{(\ell_*,\alpha_*)})} {\row_{i_*}(b^{({\ell_*},\beta_*)})} = m_{\ell_*,\alpha_*, \beta_*} (i) = \sigma'.$$

We conclude that in the case that $\ip a b = \sigma$, we have that $\ip {a'} {b'} = \sigma'$ with probability $1$, as required.

\paragraph{Soundness:}
Assume that $\ip a b \neq \sigma$, and let  $\tilde m_0$ and $\tilde m_{\ell,\alpha, \beta}$ for $\ell=1, \ldots, 10 t$ and $\alpha, \beta \in \{0,1\}$ be Merlin's messages on Step \ref{step:merlin}. We shall show that in this case Alice and Bob output $a'$, $b'$, and $\sigma'$ so that $\ip {a'} {b'} = \sigma'$ with probability at most $0.98$.

To this end, first note that we may assume that $\sum_{i=1}^{d/T} \tilde m_0(i) = \sigma$, since otherwise Alice clearly rejects on Step \ref{step:alice_check_sum}. 
On the other hand,  by (\ref{eq:complete_1}) and by assumption that $\ip a b \neq \sigma$,  we have that
 $
 \sum_{i=1}^{d/T} m_0 (i)  = \ip a b \neq \sigma.
 $
Consequently, there exists $i \in [\frac d T]$ so that $m_0 (i) \neq \tilde m_0 (i)$. 

Moreover, since $m_0(i), \tilde m_0(i) \in \{0,1,\ldots,T\}$, we have that $|m_0(i) - \tilde m_0(i)|\leq T$. Recalling that 
$t^t = T$, and that $p_\ell \geq t$ for any $\ell \in [10t]$, we conclude that at most $t$ of the $p_\ell$'s can divide $|m_0(i) - \tilde m_0(i)|$.
Thus, with probability at least $0.9$ over the choice of $\ell_*$, it holds that  $p_{\ell*}$ does not divide $|m_0(i) \neq \tilde m_0(i) |$, and so $m_0(i) \neq \tilde m_0(i) \;\;(\mathrm{mod} \;p_{\ell_*})$. In what follows, assume that this event holds.

Let $\tilde m_{\ell_*} \in (\F_{p_{\ell_*}^2})^{n_{\ell_*}}$ be given by 
$$  \tilde m_{\ell_*} = 
   \tilde   m_{\ell_*,0,0}   + ( \tilde m_{\ell_*,0,1}   + \tilde m_{\ell^*,1,0} ) \cdot X +  \tilde m_{\ell^*,1,1}   \cdot X^2   \;\;(\mathrm{mod} \;Q_{\ell_*}).$$
Next observe that we may assume that for any $i \in [\frac d T]$,
$$
  \tilde m_{\ell_*} (i)  = 
   \tilde   m_{\ell_*,0,0}  (i)  + ( \tilde m_{\ell_*,0,1} (i)  + \tilde m_{\ell^*,1,0} (i)) \cdot X +  \tilde m_{\ell^*,1,1}  (i) \cdot X^2   \;\;(\mathrm{mod} \;Q_{\ell_*}) 
=    \tilde m_0(i) \;\;(\mathrm{mod} \;p_{\ell_*}),
$$
   since otherwise Alice clearly rejects on Step \ref{step:alice_check_con}. 
On the other hand, by (\ref{eq:complete_2}) and (\ref{eq:complete_3})
 we have that for any $i \in [\frac d T]$,
$$
 m_{\ell_*} (i)  = 
   m_{\ell_*,0,0}  (i)  + (m_{\ell_*,0,1} (i)  +  m_{\ell^*,1,0} (i)) \cdot X + m_{\ell^*,1,1}  (i) \cdot X^2   \;\;(\mathrm{mod} \;Q_{\ell_*}) 
=    m_0(i) \;\;(\mathrm{mod} \;p_{\ell_*}).
$$
Consequently, by assumption that $m_0(i) \neq \tilde m_0(i) \;\;(\mathrm{mod} \;p_{\ell_*})$ for some $i \in [\frac d T]$, we have that $\tilde m_{\ell_*}(i) \neq m_{\ell_*}(i)$.

Finally, note that we may assume that $\tilde m_{\ell_*}$ is a codeword of $C_\star^{(\ell_*)}$, since otherwise Alice clearly rejects on Step \ref{step:alice_check_code}. Moreover, by (\ref{eq:complete_4}) we also have that $m_{\ell^*}$ is a codeword of $C_\star^{(\ell_*)}$. Since $C_\star^{(\ell_*)}$ has relative distance at least $0.1$, and by assumption that $\tilde m_{\ell_*} \neq m_{\ell_*}$, we have that $\tilde m_{\ell^*}$ and $m_{\ell^*}$ differ on at least a $0.1$-fraction of their entries, and so with probability at least $0.1$ over the choice of $i_*$ it holds that $\tilde m_{\ell_*} (i_*) \neq m_{\ell_*} (i_*)$. In what follows, assume that this event holds as well.

 By assumption that $\tilde m_{\ell_*} (i_*) \neq m_{\ell_*} (i_*)$, there exist $\alpha, \beta \in \{0,1\}$ so that $\tilde m_{\ell_*,\alpha,\beta}(i_*) \neq  m_{\ell_*,\alpha,\beta}(i_*)$. Consequently, with probability at least $0.25$ over the choice of $\alpha_*, \beta_*$, it holds that 
$\tilde m_{\ell_*,\alpha_*,\beta_*}(i_*) \neq  m_{\ell_*,\alpha_*,\beta_*}(i_*)$. 
But assuming that this latter event holds, we have that 
$$ \ip {a'} {b'} = \ip {\row_{i_*} (a^{(\ell_*, \alpha_*)})} {\row_{i_*} (b^{(\ell_*,  \beta_*)})} = 
m_{\ell_*,\alpha_*,\beta_*}(i_*) \neq \tilde m_{\ell_*,\alpha_*,\beta_*}(i_*) = \sigma' .$$

We conclude that in the case that $\ip a b \neq \sigma$, for any Merlin's message, we have that Alice either rejects or $\ip {a'} {b'} \neq \sigma'$ with probability at least $0.9 \cdot 0.1 \cdot 0.25 \geq 0.02$ over the choice of $\ell_*, i_*, \alpha_*$, and  $\beta_*$. So $\ip {a'} {b'} = \sigma'$ with probability at most $0.98$ over the choice of $\ell_*, i_*, \alpha_*$, and  $\beta_*$.

\section{From IP to Apx-Max-IP}\label{sec:ip_to_apx_max_ip}

In this section we use Theorem \ref{thm:ma} from the previous section which gives an MA protocol for IP over a small alphabet to give a fine-grained reduction from IP to Apx-Max-IP with a tighter polynomial dependence of the approximation parameter $\delta$ on the running time parameter $\epsilon$.

\begin{lem}[From $\IP$ to $\Apx$-$\Max$-$\IP$]\label{lem:ip_to_maxip}
The following holds for any $\epsilon >0$ and integer $c \geq 1$.
Suppose that $\IP_{N,d}$ cannot be solved in time $N^{2-\epsilon}$ for $d = c \log N$. Then $\delta$-$\Apx$-$\Max$-$\IP_{N,d'}$ cannot be solved in time $N^{2-2\epsilon}$ for $\delta = \tilde \Theta( ( \frac \epsilon c)^2 )$. 
\end{lem}

To prove the above lemma, we shall use the following encoding lemma from \cite{CW19}, which can be used to turn the (non-binary) vectors $a',b'$ from the protocol given in Theorem \ref{thm:ma} into (binary) vectors, whose inner product exhibits a gap.

\begin{lem}[Encoding Lemma, \cite{CW19}, Lemma 6.3]\label{lem:enc}
For any non-negative integers $T$ and $q$, there exist  mappings $g,h: \{0,1, \ldots, q\}^{T} \times \{0,1, \ldots, T \cdot q^2\} \rightarrow\{0,1\}^{O(T^{2}q^{4})}$ and an integer $\Gamma \leq O(T^2 \cdot q^4)$, so that for any $a,b\in \{0,1,\ldots,q\}^{T}$ and $\sigma \in \{0,1, \ldots, T \cdot q^2\}$:
\begin{itemize}
\item If $\ip  a b =\sigma \then\ip{g(a,\sigma)}{h(b,\sigma)}=\Gamma$.
\item If $\ip  a b\neq \sigma \then\ip{g(a,\sigma)}{h(b,\sigma)}<\Gamma$.
\end{itemize}
Moreover, $g,h$ can be computed in time $\poly(T,q)$. 
\end{lem}

\paragraph{The reduction:}

We shall show a reduction from  $\IP$ to many instances of $\Apx$-$\Max$-$\IP$, based on our $\ma$ protocol for $\IP$ over a small alphabet given in Theorem \ref{thm:ma}, and the above encoding Lemma \ref{lem:enc}.

Let $A,B \subseteq \{0,1\}^{d}$ and $\sigma \in \{0,1,\ldots, d\}$ be an instance of $\IP_{N,d}$. Let $T$ be a sufficiently large integer, to be determined later on, and let $\pi$ be the protocol guaranteed by Theorem \ref{thm:ma} for $T$, $q=O(\log^2T)$, and $d$ (without loss of generality assume that $T$ divides $d$).
Let $g,h$, and $\Gamma$ be the mappings and the integer guaranteed for $T$ and $q$ by Lemma \ref{lem:enc}.

Let $\rej \subseteq \{0,1\}^L$ denote the subset of Merlin's messages $m \in \{0,1\}^L$ in $\pi$ on which Alice rejects on input $\sigma$. 
For $b \in B$ and $r \in \{0,1\}^R$, let $b'_r$ denote the string output by Bob in the protocol $\pi$ on input $b$ and randomness string $r$.
Similarly, for $a \in A$ and  $r \in \{0,1\}^R$, let $a'_r$ denote the string output by Alice in the protocol $\pi$ on input $a$ and randomness string $r$. For $m \in \{0,1\}^L \setminus \rej$ and $r \in \{0,1\}^R$, let $\sigma'_{m,r}$ denote the integer output by Alice on Merlin's message $m$ and randomness string $r$.

For any $m \in \{0,1\}^L \setminus \rej$,  we create an instance $A_m,B_m$ of $\Apx$-$\Max$-$\IP_{N,d'}$, given by
$$A_m:= \{ ( g(a'_r, \sigma'_{m,r}))_{r \in \{0,1\}^R} \mid a \in A\},$$
and 
$$B_m:= \{ ( h(b'_r, \sigma'_{m,r}))_{r \in \{0,1\}^R} \mid b \in B\},$$
where 
$$ d' =O(2^R \cdot T^2 \cdot q^4).$$

Let $\delta:= \frac {0.01 \cdot 2^R } {d'}$.
Given an algorithm $\mathcal A$ for $\delta$-$\Apx$-$\Max$-$\IP_{N,d'}$, we show an algorithm $\mathcal {A'}$ for $\IP_{N,d}$:
Given an instance $A,B, \sigma$ for $\IP_{N,d}$, the algorithm $\mathcal {A'}$
generates all instances $A_m,B_m$ for $m \in \{0,1\}^L \setminus \rej$, and
runs $\mathcal A$ on any of these instances. 
If on any of the instances the algorithm $\mathcal A$ outputs a value at least $2^R \cdot (\Gamma - 0.01)$ then the algorithm $\mathcal {A'}$ accepts, otherwise it rejects.

\paragraph{Correctness: }

Correctness relies on the following claim.

\begin{claim}
$\;$
\begin{itemize}
\item If there exists $(a,b) \in A \times B$ so that $\ip a b =\sigma$, then there exist $m \in \{0,1\}^L \setminus \rej$ and  $(a'',b'') \in A_m \times B_m$ so that $\ip {a''} {b''} = 2^R \cdot \Gamma$.
\item If   $\ip a b \neq \sigma$ for any $(a,b) \in A\times B$, then $\ip {a''} {b''} \leq  2^R \cdot (\Gamma -0.02)$ for any $m \in \{0,1\}^L \setminus \rej$ and $(a'',b'') \in A_m \times B_m$.
\end{itemize}
\end{claim}

\begin{proof}
For the first item, suppose that there exists  $(a,b) \in A \times B$ so that $\ip a b =\sigma$. Let $m$ be Merlin's message in the protocol $\pi$ on inputs $a,b$, and $\sigma$, and let $(a'',b'') \in A_m \times B_m$ be given by
$a''=(g(a'_r, \sigma'_{m,r}))_{r \in \{0,1\}^R} $ and $b''=( h(b'_r, \sigma'_{m,r}))_{r \in \{0,1\}^R}$. By the completeness property of $\pi$, we have that $m \notin \rej$, and 
$\ip {a'_r} {b'_r} = \sigma'_{m,r}$ for any $r \in \{0,1\}^R$. Consequently, by Lemma \ref{lem:enc}, $\ip {g(a'_r, \sigma'_{m,r}) } 
{ h(b'_r, \sigma'_{m,r})} = \Gamma$ for any $r \in \{0,1\}^R$. But this implies in turn that
$$ \ip {a''} {b''} = \sum_{r \in \{0,1\}^R} \ip {g(a'_r, \sigma'_{m,r}) } { h(b'_r, \sigma'_{m,r})} = 2^R \cdot \Gamma.$$

For the second item, suppose that $\ip a b \neq \sigma$ for any $(a,b) \in A\times B$. Fix $m \in \{0,1\}^L \setminus \rej$ and 
$(a'',b'') \in A_m \times B_m$. Then by construction,  $a''=(g(a'_r, \sigma'_{m,r}))_{r \in \{0,1\}^R} $ and $b''=( h(b'_r, \sigma'_{m,r}))_{r \in \{0,1\}^R}$. By the soundness property of $\pi$, for at least a $0.02$-fraction of the randomness strings $r \in \{0,1\}^R$, it holds that $\ip {a'_r} {b'_r} \neq \sigma'_{m,r}$. Consequently, by Lemma \ref{lem:enc} for at least a $0.02$-fraction of the randomness strings $r \in \{0,1\}^R$, it holds that  $\ip {g(a'_r, \sigma'_{m,r}) } 
{ h(b'_r, \sigma'_{m,r})} \leq \Gamma -1$. But this implies in turn that

$$\ip {a''} {b''}  =  \sum_{r \in \{0,1\}^R} \ip {g(a'_r, \sigma'_{m,r}) } { h(b'_r, \sigma'_{m,r})}  
 \leq  0.98 \cdot 2^R \cdot \Gamma + 0.02  \cdot 2^R\cdot (\Gamma -1) 
 =   2^R \cdot (\Gamma - 0.02)  .
 $$

\end{proof}

Now, if there exists $(a,b) \in A \times B$ so that $\ip a b =\sigma$, then by the above claim there exists $m \in \{0,1\}^L \setminus \rej$ so that $\max_{a'' \in A_m, b'' \in B_m} \ip {a''} {b''} \geq 2^R \cdot \Gamma$. Consequently, the algorithm $\mathcal A$ will output a value greater than $2^R \cdot \Gamma - \delta \cdot d' = 2^R \cdot (\Gamma - 0.01)$ on the instance $A_m,B_m$, and so the algorithm $\mathcal {A'}$ will accept.

If on the other hand, $\ip a b \neq \sigma$ for any $(a,b) \in A\times B$, then by the above claim $\max_{a'' \in A_m, b'' \in B_m} \ip {a''} {b''} \leq 2^R \cdot (\Gamma - 0.02)$ for any $m \in \{0,1\}^L \setminus \rej$. Consequently, the algorithm $\mathcal A$ will output a value at most $2^R \cdot (\Gamma - 0.02)<2^R \cdot (\Gamma - 0.01)$ on any of the instances, an so the algorithm 
$\mathcal{A}'$ will reject.

\paragraph{Running time:} 
Suppose that the algorithm $\mathcal{A}$ for $\delta$-$\Apx$-$\Max$-$\IP_{N,d'}$  runs in time $N^{2-2\epsilon}$, we shall show that for an appropriate choice of $T$, the running time of the algorithm $\mathcal{A'}$ for 
$\IP_{N,d}$ is at most $N^{2-\epsilon}$. 

The algorithm $\mathcal{A'}$ enumerates over all possible Merlin's messages  $m \in \{0,1\}^L$, and for each such message checks whether Alice rejects $m$ in $\pi$, which takes time $\poly(d)$, and if she does not reject, it generates the instance $A_m, B_m$ which takes time $N \cdot 2^R \cdot \poly(d) \cdot \poly(T,q) \leq N \cdot \poly(d)$, and runs the algorithm $\mathcal A$ on $A_m, B_m$ which takes time $N^{2-2\epsilon}$. 

Hence the total running time of the algorithm $\mathcal{A}'$ is at most
\begin{eqnarray*}
 2^L \cdot (N \cdot \poly(d) + N^{2-2\epsilon}) &  \leq & 2^{O  ( \frac d T \cdot \log^2 T)} \cdot (N \cdot \poly(d) + N^{2-2\epsilon}) \\
  &  = & 2^{O (  \frac{c\log N} {T}\cdot \log^2 T)} \cdot (N \cdot \poly(c \log N) + N^{2-2\epsilon}). \\
    &  \leq & 2^{O (  \frac{c\log N} {T}\cdot \log^2 T)} \cdot N^{2-2\epsilon}. 
 \end{eqnarray*}
Finally, it can be verified that the latter expression is at most $N^{2 - \epsilon}$ for  choice of $T = \tilde \Theta( c/\epsilon) $ which divides $d$.

\paragraph{Approximation parameter:}

By choice of 
$\delta = \frac {0.01 \cdot 2^R } {d'}$, $d' = O(2^R \cdot T^2 \cdot q^4)$,  $T = \tilde \Theta( c/\epsilon) $,
and $q = O(\log^2T)$, we have that
$$\delta = \Theta \left( \frac{1} {T^2 q^4} \right) = \tilde \Theta \left(\left(\frac \epsilon c\right)^2\right).$$

\section{Applications}\label{sec:apps}

In this section we show a couple of consequences of Lemma \ref{lem:ip_to_maxip} to obtaining tighter fine-grained hardness of approximation results based on the IP assumption.

\paragraph{Closest pair in Hamming metric.}

The following reduction from Max-IP to CP in the  Hamming metric $\Delta$ is implicit in \cite{Rub18}.

\begin{lem}[From $\Apx$-$\Max$-$\IP$ to $\Apx$-$\CP_{\Delta}$, \cite{Rub18}]\label{lem:maxip_to_hamming}
Suppose that $\delta$-$\Apx$-$\Max$-$\IP_{N,d}$ cannot be solved in time $N^{2-\epsilon}$. Then 
$\delta'$-$\Apx$-$\CP_{N,d',\Delta}$ cannot be solved in time $N^{2-2\epsilon}$ for 
$\delta' = \frac \delta 2$. 
\end{lem}

The above lemma and Lemma  \ref{lem:ip_to_maxip} readily imply the following.

\begin{cor}[From $\IP$ to $\Apx$-$\CP_{\Delta}$]\label{cor:ApxCP}
Suppose that $\IP_{N,d}$ cannot be solved in time $N^{2-\epsilon}$ for $d = c \log N$. Then $\delta$-$\Apx$-$\CP_{N,d',\Delta}$ cannot be solved in time $N^{2-2\epsilon}$ for 
$\delta = \tilde \Theta( ( \frac \epsilon c)^2) $.
\end{cor}

In contrast, it is known how to obtain an $(1+\delta)$-approximation for CP over the Hamming metric in time $N^{2-\epsilon}$ for $\delta = \tilde \Theta( \epsilon^{3})$ \cite{ACW16}.

\paragraph{Closest pair in $\ell_p$ metric.}

The following reduction from CP in the Hamming metric to CP in the  $\ell_p$ metric  is also implicit in \cite{Rub18}.

\begin{lem}[From $\Apx$-$\CP_{\Delta}$ to $\Apx$-$\CP_{\ell_p}$, \cite{Rub18}] \label{lem:hamtolp}
    Suppose that $\delta$-$\Apx$-$\CP_{N, d, \Delta}$ cannot be solved in time $N^{2-\epsilon}$. Then for any $p>0$, $\delta'$-$\Apx$-$\CP_{N,d,\ell_p}$ cannot be solved in time $N^{2-2\epsilon}$ for $\delta' = \Theta_p(\delta)$.
\end{lem}

The following corollary is a consequence of the above lemma and Corollary \ref{cor:ApxCP}, and implies Theorem \ref{thm:main}.

\begin{cor}[From $\IP$ to $\Apx$-$\CP_{\ell_p}$]\label{cor:iptolp}
    Suppose that $\IP_{N,d}$ cannot be solved in time $N^{2-\epsilon}$ for $d = c \log N$. Then for any $p>0$, $\delta$-$\Apx$-$\CP_{N,d',\ell_p}$ cannot be solved in time $N^{2-2\epsilon}$ for $\delta = \tilde \Theta_p((\frac{\epsilon}{c})^{2})$.
\end{cor}

In contrast, it is known how to obtain an $(1+\delta)$-approximation for CP over the $\ell_p$ metric in time $N^{2-\epsilon}$ for $\delta=\tilde O (\epsilon^{3})$ and $p \in \{1,2\}$ \cite{ACW16}.

\paragraph{Closest pair in edit distance metric.}

For $a,b\in\Sigma^{d}$, we let $ED(a,b)$ denote their \textsf{edit distance} which is the minimum number of character deletion, insertion, and substitution operations needed to transform $a$ into $b$.   The following Lemma is also implicit in \cite{Rub18}.

\begin{lem}[From $\Apx$-$\CP_{\Delta}$ to $\Apx$-$\CP_{ED}$, \cite{Rub18}]\label{lem:CPtoED}
Suppose that $\delta$-$\Apx$-$\CP_{N, d, \Delta}$ cannot be solved in time $N^{2-\epsilon}$. Then $\delta'$-$\Apx$-$\CP_{N,d',ED}$ cannot be solved in time $N^{2-2\epsilon}$ for  $\delta' = \Theta(\delta)$.
\end{lem}

The above lemma and Corollary \ref{cor:ApxCP} imply the following corollary.

\begin{cor}[From $\IP$ to $\Apx$-$\CP_{ED}$] \label{cor:iptoed}
Suppose that $\IP_{N,d}$ cannot be solved in time $N^{2-\epsilon}$ for $d = c \log N$. Then $\delta$-$\Apx$-$\CP_{N,d',ED}$ cannot be solved in time $N^{2-2\epsilon}$ for $\delta' = \tilde \Theta( ( \frac \epsilon c)^2 ) $.
\end{cor}

To the best of our knowledge, it is not known how to solve $(1+\delta)$-$\Apx$-$\CP_{\ED}$ in sub-quadratic time.

\begin{rem}[$\Apx$-$\Min$-$\IP$ and Furthest-Pair]
	It is not hard to show (see e.g., \cite{CW19}, Lemma 5.3) that there is a simple linear-time reduction from $\delta$-$\Apx$-$\Max$-$\IP_{N,d}$ to $\delta$-$\Apx$-$\Min$-$\IP_{N,d}$ (and vice versa), and so the same result as in Lemma \ref{lem:ip_to_maxip} also holds for $\delta$-$\Apx$-$\Min$-$\IP_{N,d}$  (where the goal is to output a number in $[M,M+\delta \cdot d]$, where $M:=\min_{a \in A,b \in B} \ip a b$).
	
	Using $\Apx$-$\Min$-$\IP$ as the starting point for the reductions cited above instead of $\Apx$-$\Max$-$\IP$ implies the same results as in Corollaries \ref{cor:ApxCP}, \ref{cor:iptolp}, and \ref{cor:iptoed} for Furthest Pair (where the goal is to output a number in $[(1-\delta)M,M]$, where $M:=\max_{a \in A, b \in B}dist(a,b)$).
\end{rem}

\paragraph{Data structure setting.}
Our results extend to the data structure setting.
\begin{defn}[Approximate Nearest Neighbor ($\Apx$-$\NN$)] 
    Let $\dist:\{0,1\}^d \times\{0,1\}^d \to\field R^+$ be a distance function, and let $\delta > 0$ be a parameter. In the \textsf{Approximate Nearest Neighbor $\delta$-$\Apx$-$\NN_{N,d,dist}$ problem}, given a set $A\subseteq\{0,1\}^d$ of cardinality $N$, the goal is to pre-process the set, so that given a vector $b\in\{0,1\}^d$ it is possible to quickly output a number in $[M, (1+\delta)M]$, where $M:=\min_{a \in A} \dist(a,b).$
\end{defn}

It is known that $\Apx$-$\CP$ can be reduced to $\Apx$-$\NN$ \cite{WW18} (see also proof of Corollary 1.4 in \cite{ARW17}).

\begin{lem}[From $\Apx$-$\CP_{\Delta}$ to $\Apx$-$\NN$, \cite{WW18}] \label{lem:CPtoANN}
    Let $\dist:\{0,1\}^d \times\{0,1\}^d \to\field R^+$ be a distance function. Suppose that $\delta$-$\Apx$-$\CP_{N, d, dist}$ cannot be solved in $N^{2-\epsilon}$ time. Then for any $r>0$, $\delta$-$\Apx$-$\NN_{N,d, dist}$ cannot be solved with $N^r$ preprocessing time and $N^{1-2r\epsilon}$ time.
\end{lem}

\begin{cor}[From $\IP$ to $\Apx$-$\NN$]
Suppose that $\IP_{N,d}$ cannot be solved in time $N^{2-\epsilon}$ for $d = c \log N$. Then for any distance function $\dist \in \{\Delta,\ell_p,\ED\}$ and $r>0$, $\delta$-$\Apx$-$\NN_{N,d',dist}$ cannot be solved with $N^r$ preprocessing time and $N^{1-3r\epsilon}$ query time for $\delta = \tilde \Theta (( \frac \epsilon c)^2)$.
\end{cor}

\bibliographystyle{alpha} 
\bibliography{bibfile}

\end{document}